%% file: MainV3.tex
\documentclass[conference]{IEEEtran}
\IEEEoverridecommandlockouts
\usepackage{cite}
\usepackage{amsmath,amssymb,amsfonts}

\usepackage{graphicx}
\usepackage{dsfont}
\usepackage{textcomp}
\usepackage{xcolor}
\usepackage{amsthm}
\usepackage{algorithm}
\usepackage{algpseudocode}
\def\BibTeX{{\rm B\kern-.05em{\sc i\kern-.025em b}\kern-.08em
    T\kern-.1667em\lower.7ex\hbox{E}\kern-.125emX}}

\input{StylesMKS}

\begin{document}

\title{Optimal Singular Perturbation-based Model Reduction for Heterogeneous Power Systems\\
}

\author{\IEEEauthorblockN{Yue Huang, Dixant B. Sapkota, and Manish K. Singh}
\IEEEauthorblockA{\textit{Department of Electrical \& Computer Engineering, University of Wisconsin--Madison} \\
Emails: \{yhuang675, dbsapkota, manish.singh\}@wisc.edu}
}

\maketitle

\begin{abstract} 
Power systems are globally experiencing an unprecedented growth in size and complexity due to the advent of nonconventional generation and consumption technologies. To navigate computational complexity, power system dynamic models are often reduced using techniques based on singular perturbation. However, several technical assumptions enabling traditional approaches are being challenged due to the heterogeneous, and often black-box, nature of modern power system component models. This work proposes two singular perturbation approaches that aim to optimally identify \emph{fast} states that shall be reduced, without prior knowledge about the physical meaning of system states. After presenting a timescale-agnostic formulation for singular perturbation, the first approach uses greedy optimization to sequentially select states to be reduced. The second approach relies on a nonlinear optimization routine allowing state transformations while obtaining an optimally reduced model. Numerical studies on a test system featuring synchronous machines, inverters, and line dynamics demonstrate the generalizability and accuracy of the developed approaches.         

\end{abstract}


\section{Introduction}
The growing size and complexity of power systems are being mirrored by the rise of increasingly complex underlying mathematical models. To conduct system-level analyses, component dynamic models are often sourced from the distinct stakeholders and integrated into simulation environments. Model reduction techniques are frequently employed to reduce the computational costs of various analysis and design applications. Model reduction research in power systems began with the pursuit of dynamic equivalents of neighboring areas when simulating a study area, primarily applying coherency-based methods~\cite{chow1982timebook}. Soon followed the rigorous developments in singular perturbation (SP) approaches~\cite{KOKOTOVIC76,KOKOTOVIC1980}. With the rise in digital computing,  Krylov subspace and balanced truncation methods gained popularity~\cite{chow2013power,PaiKrylov}. Although armed with stronger guarantees on model-reduction error, these methods typically do not preserve system structure. Therefore, coherency and SP-based approaches are preferred in industry~\cite{Aranya2011measurementdynequiv}, where one seeks model interpretability, while balanced truncation is reserved for computational applications~\cite{Tari24GM}.

Singular perturbation approaches exploit the timescale separation in the dynamics of diverse physical phenomena. In such setups, the system states are first partitioned into slow and fast states. The fast states are assumed to reach quasi-steady state quickly and are eliminated, leaving a reduced-order model that accurately captures the dominant slow dynamics~\cite{KOKOTOVIC76,KOKOTOVIC1980}. The comparison of timescales within a component, say a particular generator or load, has been relatively easier than comparing across components. Therefore, SP applications in power systems typically involve the reduction of individual component models followed by integrating these to obtain a system-level model~\cite{Olaolu23TPWRSmodelreduction,saber22TCNSsingular}.

Modern grids are challenging the prerequisites for effective application of SP-- discernible and pronounced timescale separation~\cite{dominic23annualreview}. In traditional power grids, identifying and separating phenomena at different timescales was facilitated by the homogeneous, well-understood physics of a few generator and load types. However, component-level model heterogeneity has been increasing due to demand-side agility, inverter-based resources, and large loads such as data centers. Timescales for new technologies are not easily discernible as they are determined by programmable parameters, rather than physical phenomena. Furthermore, data-driven models are not amenable to physical intuition-based simplifications.

To enable the continued use of SP-based power system model reduction in modern settings, this work proposes two approaches to identify the \emph{fast} states to be reduced without relying on prior physical insights into the system timescales. To facilitate our formulations, we succinctly present a general SP formulation for linear time-invariant (LTI) systems that yields a reduced model given the choice of fast and slow states. Thereafter, we develop a greedy optimization-based approach that sequentially selects the states to be reduced. To complement the speed and accuracy aspects of the greedy approach, we formulate a nonlinear optimization routine that allows state transformations while obtaining a reduced model. We numerically evaluate the performance of the developed approaches on a test system that features dynamic models for synchronous machines, grid-following inverters, and transmission lines with heterogeneous model parameters. Our tests demonstrate the generalizability of the presented approach and provide accuracy benchmarks for future developments.  

\section{Problem Overview}
In this work, we focus on the small-signal LTI modeling of interconnected power systems hosting diverse resources. An example instance, shown in Fig.~\ref{fig:testnetwork}, includes synchronous generators (SG) and grid-following (GFL) inverter-based generations, passive constant-current loads, RL line dynamics, and shunt capacitances. While the model highlights are provided in Section~\ref{sec:tests}, the detailed model is hosted at\cite{Sapkota-SM-GFL-Model_2025}. Consider the baseline model for such a system represented as  
\begin{equation}
    \Sigma :\begin{cases}
        \dot{x} = Ax+Bu\\
        y = Cx,
    \end{cases}
    \label{eq:lti_system}
\end{equation}
where matrix $A\in \mathbb{R}^{n\times n}$ is assumed to be Hurwitz, implying small-signal stability of the baseline model. The input and output matrices are denoted as $B\in \mathbb{R}^{n\times m}$ and $C\in \mathbb{R}^{p\times n}$, respectively. The model does not have a direct feedthrough from inputs $u$ to outputs $y$.

The main objective of this work is to obtain an SP-based reduced model of a desired order $r<n$
\begin{equation}
    \hat{\Sigma}:\begin{cases}
        \dot{\hat{z}} &= \hat{A}\hat{z}+\hat{B}u\\
        \hat{y} &= \hat{C}\hat{z} +\hat{D}u,
    \end{cases}
    \label{eq:lti_system_sp}
\end{equation}
such that the input-output mapping of~\eqref{eq:lti_system_sp} closely approximates that of~\eqref{eq:lti_system}. 

\begin{figure}[t]
    \centering
    \includegraphics[width=1\linewidth]{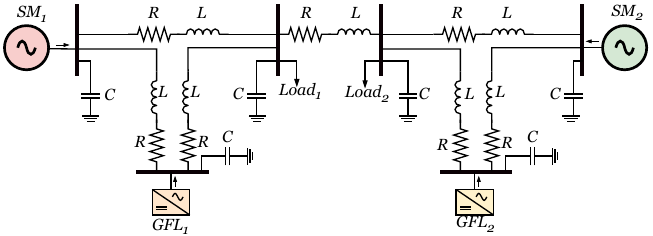}
    \caption{Test network with synchronous machine and grid-following inverter-based generation, RLC line models, and passive constant current loads~\cite{Sapkota-SM-GFL-Model_2025}.}
    \label{fig:testnetwork}
\end{figure}
\section{Model Reduction Setup}

\subsection{Singular Perturbation Preliminaries}\label{sec:SP}
To cohesively present our approaches for optimally selecting the states to be reduced via SP, we build on the formalism of~\cite{structured_model_reduction_takayuki}. Without loss of generality, consider a state transformation $z=Ux$, where $U\in\mathbb{R}^{n\times n}$ is a unitary matrix partitioned as $U=[P^\top~Q^\top]^\top$. Accordingly, partition the transformed states as $z=[z_p^\top~z_q^\top]^\top$ such that $z_p=Px\in\mathbb{R}^r$ denotes the states to be retained after model reduction. Notice that the prevalent setting where the slow and fast states are identified without state transformations can be captured here by setting $(P,~Q)$ to be selection matrices. For $U$ to be a unitary matrix, the following must hold
\begin{equation}\label{eq:PQ_cons}
    PP^\top = I_{r}, ~~QQ^\top = I_{n-r}, ~~PQ^\top = 0.
\end{equation} 
Applying the transformation $x = U^\top z$, LTI~\eqref{eq:lti_system} can be equivalently written as
\begin{subequations}\label{eq:lti_system_hat}
\begin{align}
        \begin{bmatrix}
            \dot{z_p}\\ \dot{z_q}
        \end{bmatrix} &= 
        \begin{bmatrix}
            PAP^\top ~~PAQ^\top\\
            QAP^\top ~~QAQ^\top
        \end{bmatrix} 
        \begin{bmatrix}
            z_p\\z_q
        \end{bmatrix} +
        \begin{bmatrix}
            PB\\QB
        \end{bmatrix}u\label{seq:LTIa}
        \\
        y &= \begin{bmatrix}
            CP^\top ~~CQ^\top
        \end{bmatrix}
        \begin{bmatrix}
            z_p\\z_q
        \end{bmatrix}
\end{align}
\end{subequations}
When matrices $(P, Q)$ are selected to yield $z_p$ and $z_q$ as slow and fast states, the SP approximation can be obtained by setting $\dot{z}_q=0$ in the second row equation of~\eqref{seq:LTIa} to obtain the algebraic approximation
\begin{equation}\label{eq:z_q_approx}
    \hat{z}_q = -(QAQ^\top)^{-1}QAP^\top\hat{z}_p-(QAQ^\top)^{-1}QBu,
\end{equation}
where $\hat{z}_q$ is the quasi-steady-state approximation of $z_q$ and $\hat{z}_p$ is the SP-approximated state corresponding to $z_p$. The overall SP-based reduced model is then obtained by substituting~\eqref{eq:z_q_approx} back in~\eqref{eq:lti_system_hat}. To map the reduced model to~\eqref{eq:lti_system_sp}, we relabel $\hat{z}_p$ to $\hat{z}$ and obtain the state-space matrices as
\begin{equation}\label{eq:hat_matrices}
    \begin{aligned}
        \hat{A} &:= PAP^\top-PA\Pi AP^\top &&\in \mathbb{R}^{r\times r}\\
        \hat{B} &:= (P-PA\Pi)B  &&\in\mathbb{R}^{r\times m}\\
        \hat{C} &:=C(P^\top-\Pi AP^\top) &&\in \mathbb{R}^{p\times r}\\
        \hat{D}&:= - C\Pi B &&\in\mathbb{R}^{p \times m}\\
    \end{aligned}
\end{equation}
where
\begin{equation}\label{eq:pi_equation}
\begin{aligned}
         \Pi &:=Q^\top(QAQ^\top)^{-1}Q &&&\in \mathbb{R}^{n\times n}.
\end{aligned}
\end{equation}

\subsection{Desirable reduced model properties}\label{sec:desirable}
Having stated the general SP formulation, the goal of this work boils down to optimally selecting matrices $(P, Q)$ that render~\eqref{eq:lti_system_sp} as a good order-$r$ approximation of the original dynamics~\eqref{eq:lti_system} while ensuring certain desirable properties. In particular, to ensure stability of~\eqref{eq:lti_system_sp}, we desire $\hat{A}$ to be Hurwitz. Further, to avoid a direct feedthrough, we will pursue $(P, Q)$ that result in $\hat{D}=0$. Note from~\eqref{eq:hat_matrices} that $\hat{D}=0$ if $C\Pi=0$ or $\Pi B=0$. Since $PQ^\top=0$ from~\eqref{eq:PQ_cons}, one obtains $C\Pi=0$ if $\range(C^\top)\subseteq\range(P^\top)$. Similarly, $\Pi B=0$ if $\range(B)\subseteq\range(P^\top)$. Thus, ensuring $\hat{D}=0$ enforces $\range(P^\top)$ to contain either $\range(C^\top)$ or $\range(B)$, or both. Since for our test setup, the number of outputs $p$ is less than inputs $m$, we choose the dimensionally less restrictive constraint
\begin{equation}\label{eq:CPconstraint}
    \range(C^\top)\subseteq\range(P^\top).
\end{equation}

Enforcing the above constraints further allows us to use the prevalent metric of $\mcH_2$ norm to quantify the model-reduction accuracy. Specifically, the dynamics of output error between~\eqref{eq:lti_system} and \eqref{eq:lti_system_sp} is given by
\begin{equation}\label{eq:errordyn}
    \Sigma_{\mathrm{err}}:\begin{cases}
        \begin{bmatrix}
    \dot{x}\\
    \dot{\hat{z}}
\end{bmatrix} &=
\underbrace{
\begin{bmatrix}
A & 0 \\[2pt]
0 & \hat{A}
\end{bmatrix}
}_{=\bar{A}}
\begin{bmatrix}
    x\\
    \hat{z}
\end{bmatrix}
+
\underbrace{
\begin{bmatrix}
B \\[2pt]
\hat{B}
\end{bmatrix}
}_{=\bar{B}}
u,\\
        \delta &= \underbrace{\begin{bmatrix}
            C ~~-\hat{C}
        \end{bmatrix}}_{=\bar{C}}
        \begin{bmatrix}
        x\\
        \hat{z}
        \end{bmatrix},
    \end{cases}
\end{equation}
where $\delta=y-\hat{y}$ is the output error. Note that the error dynamics in~\eqref{eq:errordyn} is exponentially stable since the block diagonal matrix $\bar{A}$ is composed of $A$ that is Hurwitz by assumption, and $\hat{A}$, which will be designed to be Hurwitz. Further, the absence of a feedthrough path in $\Sigma_{\mathrm{err}}$ results in a well-defined $\|\Sigma_{\mathrm{err}}\|_{\mcH_2}$ that can be computed as
\begin{equation}\label{eq:H2def}
\|\Sigma_{\mathrm{err}}\|_{\mcH_2}=\trace(\bar{B}^\top\Phi\bar{B}),
\end{equation}
where matrix $\Phi$ is the observability gramian of $\Sigma_{\mathrm{err}}$ that uniquely solves
\begin{equation}\label{eq:lyapunov_observability}
    \bar{A}^\top\Phi~+\Phi\bar{A} ~+\bar{C}^\top\bar{C} = 0.
\end{equation}
The way norm $\|\Sigma_{\mathrm{err}}\|_{\mcH_2}$ relates to $\|\delta(t)\|_2^2$ is amenable to various interpretations~\cite{gayme}:
	\begin{enumerate}
		\item[\emph{i1)}] The expected steady-state error between the outputs of~\eqref{eq:lti_system} and~\eqref{eq:lti_system_sp} for unit-variance white noise input $u(t)$:
			\begin{equation*} 
		\|\Sigma_{\mathrm{err}}\|_{\mcH_2}=\lim_{t\rightarrow \infty}\mathbb{E}\left[\|\delta(t)\|_2^2\right]
		\end{equation*}
		\item[\emph{i2)}] The sum of time-integrals of error $\delta(t)$ for unit-impulse disturbances applied separately for for the $m$ inputs:
			\begin{equation*}
		\|\Sigma_{\mathrm{err}}\|_{\mcH_2}=\sum_{i=0}^{m}\int_{0}^{\infty} \|\delta^i(t)\|_2^2 \mathrm{d}t,
		\end{equation*}
        where $\delta^i(t)$ denotes the output of $\Sigma_{\mathrm{err}}$ when an impulse is applied at the $i^{th}$ input.
	\end{enumerate}
Since $\|\Sigma_{\mathrm{err}}\|_{\mcH_2}$ captures the model-reduction error across broad input conditions, we use it as our performance metric to optimize. 
\section{Optimal SP-based Model Reduction}\label{sec:OptSP}
Given the full-order model~\eqref{eq:lti_system}, the reduced model~\eqref{eq:lti_system_sp} is uniquely determined by the matrices $(P, Q)$; cf.~\eqref{eq:hat_matrices}. For the prevalent SP setting where slow and fast states are selected from the original states $x$ without state transformation, matrices $(P, Q)$ are selection matrices built by partitioning the rows of an identity matrix. For this special case, the reduced model matrices~\eqref{eq:hat_matrices} are uniquely determined by either $P$ or $Q$ since $[P^\top Q^\top]^\top$ is a row permutation of the identity matrix. We next present a greedy approach to optimally select the states to be reduced under this special case.

\subsection{Greedy optimization-based singular perturbation}\label{sec:greedy}
We build the greedy approach by sequentially determining which states to reduce, i.e., selecting $n-r$ distinct canonical vectors to include in $Q$, i.e., $Q^\top=[\{e_i\}_{i\in\mcR}]$, where $e_i$ denotes the $i^{\text{th}}$ canonical vector, and  $\mcR$ denotes the set of reduced states with $|\mcR|=n-r$. In doing so, we need to account for the two desirable properties discussed in Section~\ref{sec:desirable}: stability of $\hat{A}$, and obtaining $\hat{D}=0$ by enforcing~\eqref{eq:CPconstraint}. The stability requirement can be verified by analyzing the spectrum of $\hat{A}$ at each greedy optimization step. On the other hand, aiming for $\hat{D}=0$ restricts the choice of canonical vectors that can be included in $Q$. Specifically, from~\eqref{eq:PQ_cons}, it follows that
\begin{equation}\label{eq:QC}
    \range(C^\top)\subseteq\range(P^\top)\implies QC^\top=0.
\end{equation}
Therefore, the sparsity of $C$ dictates the set of states, $\bar{\mcR}$, that can be reduced while satisfying~\eqref{eq:QC}. Particularly, $\bar{\mcR}=\{j|C^\top e_j=0\}.$ The greedy algorithm proceeds by augmenting matrix $Q$ at each step by selecting the row from $\bar{\mcR}$ that results in minimum $\|\Sigma_{\mathrm{err}}\|_{\mcH_2}$ and terminates after $n-r$ steps. An early termination can be caused due to exhaustion of the set of candidate states to be reduced $\bar{\mcR}$, or instability being caused with all available reduction options. Algorithm \ref{greedy_algorithm} summarizes the above approach.

\begin{algorithm}
    \begin{algorithmic}[hbt!] 
    \caption{Greedy optimization-based singular perturbation}\label{greedy_algorithm}
        \Require Full-order model $\Sigma$ and reduced model order $r$
        \State Initialize candidate set $\bar{\mcR}\gets \{j|C^\top e_j=0\}.$
        \State Initialize $Q_{\mathrm{opt}}\gets [~]$, stability indicator $\sigma=1$ .
        \While{$\mathrm{rows}(Q_{\mathrm{opt}})<n-r$ \& $\sigma=1$ \& $\bar{\mcR}\neq\phi$}
        \State Set $err\gets \infty$, $S\gets 0$, $\sigma\gets 0$.
        \For{$j\in\bar{\mcR}$}
        \State Augment $Q^\top\gets [Q_{\mathrm{opt}}^\top~~e_j]$.
        \State Find $P$ by stacking canonical vectors not in $Q$.
        \State Evaluate $(\hat{A}, \hat{B}, \hat{C})$ using~\eqref{eq:hat_matrices}.
        \If{$\hat{A}$ is Hurwitz}
        $\sigma\gets 1$.
        \State Evaluate $\|\Sigma_{\mathrm{err}}\|_{\mcH_2}$ using~\eqref{eq:errordyn}--\eqref{eq:lyapunov_observability}.
        \If{$\|\Sigma_{\mathrm{err}}\|_{\mcH_2}<err$}
        $err\gets \|\Sigma_{\mathrm{err}}\|_{\mcH_2}$, \State $S\gets j$.
        \EndIf
        \EndIf
        \EndFor
        \If{$\sigma=1$} $Q_{\mathrm{opt}}^\top\gets [Q^\top~~e_S]$, $\bar{\mcR}\gets \bar{\mcR}\setminus S$. 
        \EndIf
        \EndWhile
    \end{algorithmic}
\end{algorithm}

\subsection{Nonlinear optimization-based singular perturbation}\label{sec:nonlinear}
In this section, we relax the requirement for $(P, Q)$ to be selection matrices and pursue optimal model reduction while allowing state transformations. It is conceivable that such relaxation could improve the accuracy of reduced models. Moreover, relaxing the combinatorial constraints allows us to leverage nonlinear programming solvers to obtain a (locally) optimal SP-based reduction. However, we acknowledge that state transformations may diminish the physical interpretability of states and lead to a loss of structure. Therefore, we envisage that such approaches will be useful for reducing data-driven or black-box baseline models. The presented approach will also enable future research on preserving structure in parts of a power system interconnection while using state transformations in the remainder to improve accuracy.

When $[P^\top Q^\top]$ is allowed to be any unitary matrix, it is worth noting that the reduced model matrices $(\hat{A}, \hat{B}, \hat{C})$ depend solely on the choice of matrix $P$~\cite{structured_model_reduction_takayuki}. This is because matrix $Q$ features in~\eqref{eq:hat_matrices} via matrix $\Pi$, which is invariant to the basis selection of $Q$. Unfortunately, even though matrix $A$ is Hurwitz, the conditions on $P$~\eqref{eq:PQ_cons} do not necessarily yield a Hurwitz $\hat{A}$. However, the flexibility of state transformation allows us to guarantee the stability of the reduced model conveniently. Specifically, since $A$ is Hurwitz, we can solve $A^\top X+XA=-I$ to obtain a unique symmetric positive-definite matrix $X$ with a Cholesky decomposition $X=LL^\top$. The state transformation $\Tilde{x}:=L^\top x$ of original dynamics~\eqref{eq:lti_system} yields an equivalent representation   
\begin{equation}
    \Tilde{\Sigma} :\begin{cases}
        \dot{\Tilde{x}} = \underbrace{L^\top AL^{-\top}}_{=\Tilde{A}}\Tilde{x}+\underbrace{L^\top B}_{=\Tilde{B}}u\\
        y = \underbrace{CL^{-\top}}_{=\Tilde{C}}\Tilde{x},
    \end{cases}
    \label{eq:lti_tilde}
\end{equation}
where matrix $\Tilde{A}$ is negative definite ($v^\top\Tilde{A}v<0,~\forall v\neq 0)$, though not symmetric. Applying the SP steps of Section~\ref{sec:SP} guarantees that the resulting $\hat{A}$ is negative definite, hence Hurwitz~\cite[Lemma 2]{structured_model_reduction_takayuki}.

Next, enforcing~\eqref{eq:CPconstraint} for $\Tilde{\Sigma}$ can be done by partitioning $P$ as
\begin{equation}\label{eq:partitionP}
    P=\begin{bmatrix}
        P_{\mathrm{fix}}\\ R,
    \end{bmatrix}
\end{equation}
where $P_{\mathrm{fix}}^\top$ is an orthogonal basis for $\range(\Tilde{C}^\top)$ that can be obtained via Gram-Schmidt process. The optimal model reduction task then narrows down to determining $R\in\mathbb{R}^{(r-p)\times n}$ that satisfies the orthonormality conditions
\begin{equation}\label{eq:orthogonalR}
    RR^\top=I~\text{and}~RP_{\mathrm{fix}}^\top=0.
\end{equation}
Since $P_{\mathrm{fix}}$ is known a priori, constraint~\eqref{eq:orthogonalR} can be further simplified by obtaining the orthogonal complement $V\in\mathbb{R}^{(n-p)\times n}$ of $P_{\mathrm{fix}}^\top$, i.e. $VP_{\mathrm{fix}}^\top=0$ and $VV^\top=I$, and restricting  
\begin{equation}\label{eq:RWV}
    R=WV,
\end{equation}
where $W\in\mathbb{R}^{(r-p)\times (n-p)}$ is our ultimate decision variable for model reduction.
\begin{lemma}\label{le:Worthonormal}
    Given Eq.~\eqref{eq:partitionP} and \eqref{eq:RWV}, the orthonormality requirement of~\eqref{eq:PQ_cons}, $PP^\top=I_r$ is equivalent to $WW^\top = I_{r-p}$.
\end{lemma}
\begin{proof}
Substituting $R = WV$ in~\eqref{eq:partitionP} gives
\begin{equation}
PP^\top = 
\begin{bmatrix}
P_{fix} \\ WV
\end{bmatrix}
\begin{bmatrix}
P_{fix}^\top & V^\top W^\top
\end{bmatrix}
=
\begin{bmatrix}
I_k & 0 \\ 0 & WW^\top
\end{bmatrix},
\end{equation}
where we used $VP_{\mathrm{fix}}^\top=0$ and $VV^\top=I$.
Therefore, the orthogonality constraint $PP^\top = I_r$ is equivalent to $WW^\top = I_{r-p}$.
\end{proof}

To recapitulate, given a matrix $W$ satisfying $WW^\top=I$, one reduces $\Tilde{\Sigma}$ to uniquely determine a reduced model~\eqref{eq:lti_system_sp} where the matrices $(\hat{A}, \hat{B}, \hat{C})$ are given by~\eqref{eq:hat_matrices}, \eqref{eq:partitionP}, and~\eqref{eq:RWV}. Acknowledging that the mapping $W\mapsto(\hat{A}(W), \hat{B}(W), \hat{C}(W)$ leads to a well defined mapping $W\mapsto(\bar{A}(W), \bar{B}(W), \bar{C}(W)$ via~\eqref{eq:errordyn}, the optimal model reduction task can be formulated as
\begin{subequations}\label{eq:P1}
\begin{align}
\min_{W\in\mathbb{R}^{(r-p)\times(n-p)}}~&~\trace(\bar{B}^\top\Phi\bar{B})\qquad\qquad\textrm{(P1)}\notag\\
\textrm{s.to}~&~\eqref{eq:hat_matrices},~\eqref{eq:pi_equation},~\eqref{eq:lyapunov_observability}\\
~&~WW^\top=I,\label{seq:pg}
\end{align}
\end{subequations}
where the objective captures the error norm $\|\Sigma_{\mathrm{err}}\|_{\mcH_2}$ as per~\eqref{eq:H2def}. We solve (P1) using off-the-shelf nonlinear solvers. The implementation details and numerical results on the model-reduction accuracy obtained by implementing Algorithm~\ref{greedy_algorithm} and solving (P1) for varying values of the reduced model $r$ are provided next.

\section{Numerical Tests}\label{sec:tests}
We evaluated the proposed model reduction approaches by conducting numerical experiments on the test setup shown in Fig.~\ref{fig:testnetwork}. The test model is made available at~\cite{Sapkota-SM-GFL-Model_2025}; a brief overview is provided next.

\subsection{Test setup}
The six-bus test setup shown in Fig.~\ref{fig:testnetwork} consists of two SGs, two grid-following inverters, and two constant-current demands. The power lines are modeled as series resistance and inductance with shunt capacitance at both ends; the $RLC$ values for all lines were chosen to be identical. All dynamics were modeled in $\mathrm{dq}$ reference frame using per unit quantities~\cite{per_unit_inverter_modeling}. The model inputs included the inverter and load current references, the SG active power, and the field-circuit voltage references, yielding $u\in\mathbb{R}^{12}$. The observed output was the frequency at the two SG terminals, implying output dimension $p=2$. 

Each SG was represented by a $9^{\text{th}}$ order model that captures the swing, governor, stator current, field current, and damper winding dynamics~\cite{kundur_stability_book}. The baseline model parameters were taken from~\cite[Example 4.1]{kundur_stability_book}. While the per-unit electrical parameters of the two SGs were kept identical, the mechanical parameters (inertia, droop, and governor time constants) were distinguished to emulate heterogeneity; see~\cite{Sapkota-SM-GFL-Model_2025} for details. A $6^{\text{th}}$ order model for GFL inverters from~\cite{per_unit_inverter_modeling} was used. This model features the current and frequency dynamics governed by a current controller and a phase-locked loop. Including the $RLC$ dynamics of the power lines resulted in a total of $n=56$ states. The overall model was linearized at the equilibrium to obtain the LTI representation~\eqref{eq:lti_system}.

\subsection{Model reduction results}
For the above test system with 12 inputs, 56 states,and 2 outputs, we ran the greedy optimization-based Algorithm~\ref{greedy_algorithm} for decreasing value of the reduced model order. The requirement~\eqref{eq:CPconstraint} uniquely determines the permissible $2^{\text{nd}}$ order model and does not allow further reduction. However, we observed that Algorithm~\ref{greedy_algorithm} terminates at $r=4$ because any further reduction leads to instability. Figure~\ref{fig:results} plots the error metric $\|\Sigma_{\mathrm{err}}\|_{\mcH_2}$ for varying order of the reduced order model. While the error generally increases with decreasing model order, as expected, the relation is interestingly not strictly monotonic. For instance, the reduced $15^\text{th}$ order is more accurate than the $16^\text{th}$ order model, defying the presumption that more detailed models are more accurate. We next observed the sequence of states being reduced as Algorithm~\ref {greedy_algorithm} proceeds to lower-order models. Despite not being explicitly informed by the model's physical interpretation, the \emph{general} trend follows intuition: GFL states were among the first to be reduced, followed by network states, SG's electrical states, and finally the mechanical states of SG. However, nontrivial interleaving was also observed. For instance, some network states are retained even after all the GFL states and some electrical states have been reduced at $r=15$. 

We next solved (P1) to obtain optimally reduced models while allowing state transformations. Unlike the greedy algorithm, which sequentially reduces states, problem (P1) must be solved independently for each value of $r$. Not focusing on computational advancements, we used the MATLAB-based nonlinear optimization solver $\mathrm{fmincon}$ to solve (P1) while vectorizing the decision variable $W$, and associated constraints. We observed that the numerical complexity and solution quality of (P1) can be significantly improved when $\mathrm{fmincon}$ is initialized with the solution of the greedy optimization for the respective value of $r$. Since the greedy optimization is carried out on the original dynamics~\eqref{eq:lti_system} while (P1) is solved for the transformed dynamics~\eqref{eq:lti_tilde}, the initialization process requires the alignment steps described next. For a given $r$, let $Q_{\mathrm{opt}}$ denote the solution of Algorithm~\ref{greedy_algorithm}. Find the corresponding matrix $P_\mrg\in\mathbb{R}^{r\times n}$ satisfying $Q_{\mathrm{opt}}^\top P_\mrg=0$ and $P_\mrg P_\mrg^\top=I$. Thus, the greedy algorithm suggests retaining the states in $\range(P_{\mrg}^\top)$. Since model~\eqref{eq:lti_tilde} is obtained from~\eqref{eq:lti_system} by applying the transformation $\Tilde{x}=L^\top x$, the states retained by the greedy algorithm in the transformed state space would be $\range(L^{-1}P_{\mrg}^\top)$. Since $P_{\mathrm{fix}}$ is pre-included in $P$ to satisfy~\eqref{eq:CPconstraint} (see~\eqref{eq:partitionP}), we desire $\range([P_{\mathrm{fix}}^\top~R^\top])=\range(L^{-1}P_{\mrg}^\top)$ at initialization. Given $P_{\mathrm{fix}}, L,$ and $P_\mrg$, we can find a matrix $R$ with orthonormal rows that satisfies the range criteria. Finally, we find the initialization of (P1) as $W_{\textrm{init}}=RV^\top$, where $V$ can be found as the matrix spanning the space orthonormal to $\range(P_{\mathrm{fix}})$; see the discussion before~\eqref{eq:RWV}.

\begin{figure}[H]
    \centering
    \includegraphics[width=0.9\linewidth]{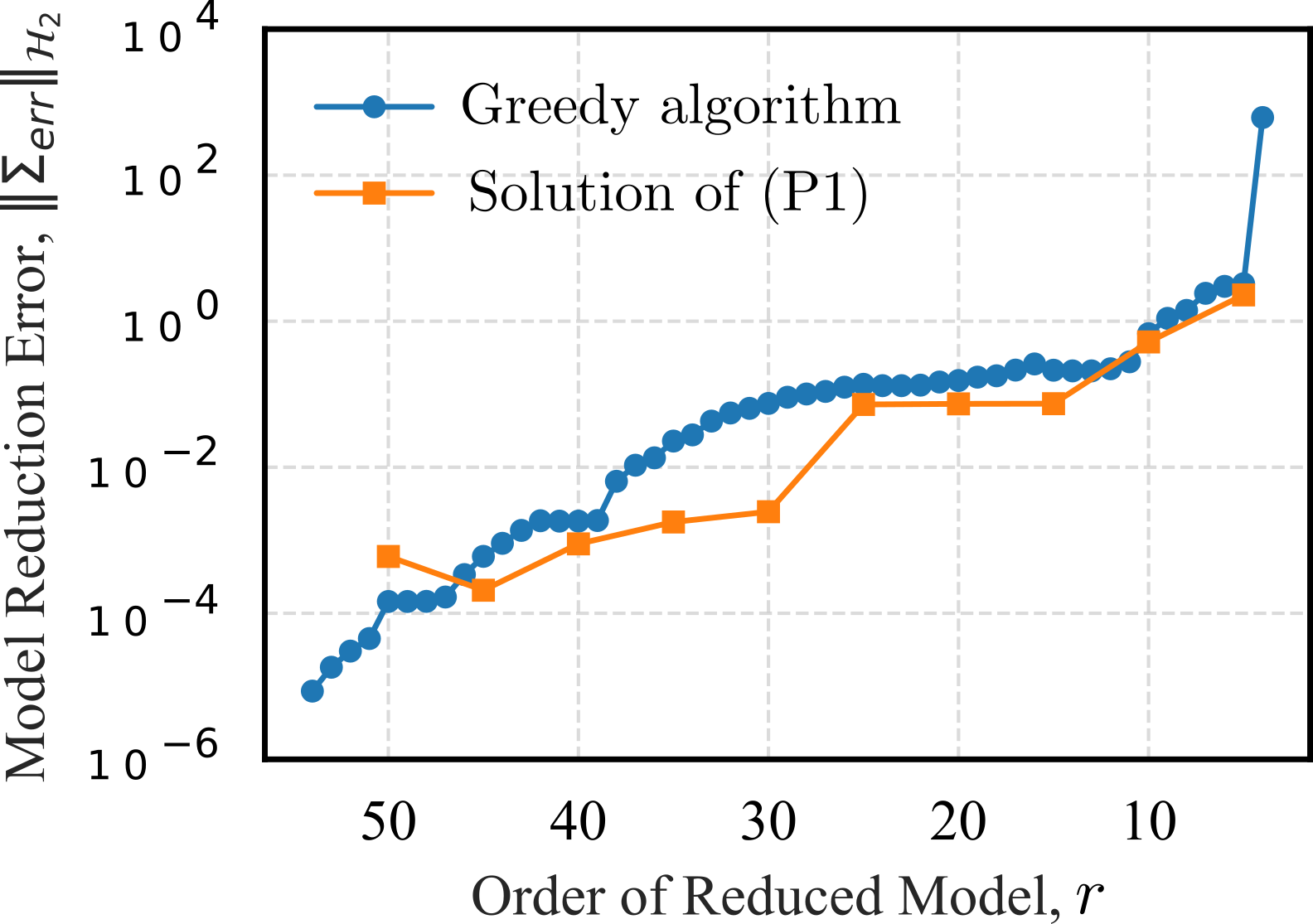}
    \caption{Approximation error $\|\Sigma_{\mathrm{err}}\|_{\mcH_2}$ attained by models reduced by greedy and nonlinear optimization methods for varying order of reduced model.}
    \label{fig:results}
\end{figure}
Figure~\ref{fig:results} shows the reduction error attained by the optimal model obtained from (P1) for $r=5,10,\dots,50$. Note that (P1) is a nonconvex problem solved by a local solver and terminated early while observing the improvements over iterates. Therefore, while for several values of $r$, the error of reduced models is lower than the result of greedy optimization by a factor of 5 to 10, the improvement is not consistent. While the numerical complexity of (P1) increases with $r$ and hence the performance could decline, Algorithm~\ref{greedy_algorithm} performs better for those instances. Therefore, the two developed approaches provide complementary value, with greedy optimization suited for reducing a few states, while nonlinear optimization is effective for aggressively reducing model order.

\section{Conclusions and Future Work}\label{sec:conclusion}
This work advances the application of singular perturbation-based power system model reduction to settings where the timescale separation of distinct state dynamics is not known a priori. After presenting a generalized framework for SP-based model reduction of LTI systems, two approaches towards optimal selection of reducible states are put forth. The first approach builds on greedy optimization techniques that sequentially identify the states that yield the minimum one-step model-reduction error. The second approach allows state transformations on the original model before applying singular perturbation, thus increasing modeling flexibility and reducing approximation error. Moreover, the second approach allows for the use of nonlinear optimization solvers that yield lower suboptimality than the greedy approach. Numerical tests of the proposed approaches demonstrate the generalizability of their application and quantify the scope of accuracy improvements arising from permitting state transformations. Improving computational tractability and adapting our nonlinear optimization-based approach to preserve desired structural properties constitute our future research.  

\bibliography{myabrv,SPreferences}
\bibliographystyle{IEEEtran}

\end{document}

%% file: StylesMKS.tex
\DeclareMathOperator{\range}{range}

\DeclareMathOperator{\trace}{Tr}

\newtheorem{lemma}{Lemma}

\newcommand \mrg{\mathrm{g}}

\newcommand \mcH{\mathcal{H}}

\newcommand \mcR{\mathcal{R}}

